\documentclass[12pt,reqno]{amsart}
\usepackage{amsmath,amsfonts,amsthm,amssymb,mathabx,dsfont}
\newcommand\version{June 11, 2018}
\usepackage{amsxtra}
\usepackage{comment,cite}
\usepackage[utf8]{inputenc}
\usepackage{mathtools}

\newtheorem{theorem}{Theorem}

\newtheorem{corollary}{Corollary}

\theoremstyle{definition}

\theoremstyle{remark}

\numberwithin{equation}{section}

\renewcommand{\epsilon}{\varepsilon}

\newcommand{\R}{\mathbb{R}}

\DeclareMathOperator{\infspec}{inf\, spec}

\newcommand{\norm}[1]{\left\| #1 \right\|}

\newcommand{\mcrit}{{m^\ast}}
\newcommand{\mcritMS}{m_{1}}
\newcommand{\mcritCur}{m_{2}}

\newcommand{\lowerBound}{0.58}
\newcommand{\upperBound}{1.73}
\newcommand{\hxi}{\hat\xi}

\begin{document}

\title{Stability of the $2+2$ fermionic system with point interactions}

\author{Thomas Moser, Robert Seiringer}
\address{Institute of Science and Technology Austria (IST Austria), Am Campus 1, 3400 Klos\-ter\-neuburg, Austria}
\email{tmoser@ist.ac.at, rseiring@ist.ac.at}


\begin{abstract} We give a  lower bound on the ground state energy of a system of two fermions of one species interacting with two fermions of another species via point interactions. We show that there is a critical mass ratio $\mcritCur \approx \lowerBound$ such that the system is stable, i.e., the energy is bounded from below, for $m \in [\mcritCur, \mcritCur^{-1}]$. So far it was not known whether this $2+2$ system exhibits a stable region at all or whether the formation of four-body bound states  causes an unbounded spectrum for all mass ratios, similar to the Thomas effect. Our result gives further evidence for the stability of the more general $N+M$ system.
\end{abstract}

\date{\version}

\maketitle


\section{Introduction}

Systems of particles interacting via point interactions are frequently used in physics to model short range forces. In these models  the shape of the interaction potential enters only via the  scattering length. Originally point interactions were  introduced in the 1930s to model nuclear interactions \cite{Bethe1935, Bethe1935B, Fermi1936, Thomas1935, wigner1933streuung},  and  later they were also  successfully applied to other areas of physics like polarons (see \cite{massignan2014polarons} and references there) or cold atomic gases \cite{ZwergerBCS}. 

Given $N \ge 1$ fermions of one type with mass $1/2$ and $M \ge 1$  fermions of another type with mass $m/2 > 0$, point interaction models give a meaning to the formal expression
\begin{equation}
-  \sum_{i=1}^N \Delta_{x_i} - \frac 1 {m} \sum_{j=1}^M \Delta_{y_j} + \gamma \sum_{i=1}^N \sum_{j=1}^M \delta(x_i-y_j) \label{eq:pointInter}
\end{equation}
for $\gamma \in \R$.
Because of the existence of discontinuous functions in $H^1(\R^n)$ for $n \geq 2$, this expression is  ill-defined in dimensions larger than one. In the following we restrict our attention to the three-dimensional case but we note that the system also exhibits interesting behavior in two dimensions \cite{DellAntonio1994, Dimock2004, Ulrich2017}. 

A mathematically precise version of \eqref{eq:pointInter} in three dimensions was constructed in \cite{DellAntonio1994, Finco2010} and we will work here with the model introduced there. 
We note that even though these models are mathematically well-defined it is not established whether they can be obtained as a limit of genuine Schr\"odinger operators with interaction potentials of shrinking support. (See, however, \cite{albeverio1988solvable} for the case $N=M=1$, and \cite{basti} for models in one dimension.)

It was already known to Thomas \cite{Thomas1935} that systems with point interactions  are inherently unstable for bosons, in the sense that the energy is not bounded from below, if there are at least three particles involved. It turns out that in the case that the particles are fermions the question of stability is more delicate as it depends on the mass ratio of the two species, in general. 

The case  $N=M=1$ is completely understood as it reduces to a one particle problem  
\cite{albeverio1988solvable}. In this case there exists a one-parameter family of Hamiltonians describing point interactions parameterized by the inverse scattering length, and they are bounded from below for all masses.

Beside this trivial case also the $2+1$ case (i.e., $N = 2$ and $M = 1$), where the two particles of the same species are fermions, is  
well understood \cite{DellAntonio1994, sherm, correggi2012stability, Correggi2015, Correggi2015B,  Becker2017, Minlos2011, Minlos2012, Minlos2014, Minlos2014B}. There is a critical mass ratio $\mcrit \approx 0.0735$ such that the system is unstable for $m < \mcrit$ and stable otherwise. It is remarkable that this critical mass ratio does not depend on the strength of the interaction, i.e., the scattering length. Recently in \cite{Becker2017} the spectrum of the $2+1$ system was discussed in more detail. Moreover, it was shown in \cite{Correggi2015, Minlos2014} that in a certain mass range other models describing point interactions can be constructed.

For larger systems of fermions even the question of stability is generally open. In \cite{MoserSeiringer2017} the stability result for the $2+1$ case was recently extended to the general $N+1$ problem ($N 
\geq 2$ and $M = 1$). In particular it was shown that there exists a critical mass  $\mcritMS \approx 0.36$ such that the system's energy is bounded from below, uniformly in $N$, for $m \ge \mcritMS$. As a consequence of the $2+1$ case this $N+1$ system is  unstable for $m < \mcrit$, but the behavior for $m \in [\mcrit,\mcritMS)$ is unknown. 

By separating particles one can  obtain an upper bound on the ground state energy of the  general $N+M$ problem using the bounds for the $N+1$ or the $1+M$ problem. We note that the latter  is, up to an overall  factor, equivalent to the $M+1$ problem with $m$ replaced by its inverse. Hence the fact that $\mcritMS < 1$ gives hope that there exists a mass region where the general $N+M$ system is stable for all $N$ and $M$. 
The simplest problem of this kind is the $2+2$ case. So far there are only numerical results on its stability available \cite{michelangeli2016,Castin2015}. In particular, the analysis  in \cite{Castin2015} suggests that the critical mass for the $2+2$ case should be equal to $\mcrit$, i.e., the one for the $2+1$ case.

In this paper we give a rigorous proof of stability for the $2+2$ system in a certain window of mass ratios. We find a critical mass $\mcritCur \approx \lowerBound$ such that the system is stable if $m \in [\mcritCur, \mcritCur^{-1}] \approx [\lowerBound,\upperBound]$. 
We note that the critical mass $\mcritCur$ is not optimal and we cannot make any further statements about the mass range $[\mcrit, \mcritCur] \cup [\mcritCur^{-1}, \mcrit^{-1}]$. The behavior for these masses, and in particular the question whether $\mcritCur = \mcrit$,  still  represents an open problem.

\section{The model}
For $p_1,p_2,k_1,k_2 \in \R^3$ and  $m>0$, let 
\begin{equation}
h_0(p_1,p_2,k_1,k_2) = p_1^2 + p_2^2 + \frac 1 m \left(k_1^2+k_2^2\right) \,.
\end{equation}
We will work with the quadratic form $F_\alpha$ introduced in \cite{Finco2010} for $2+2$ particles. Its form domain is given by
\begin{equation}
D(F_\alpha) = \{\psi =  \varphi + G_\mu \xi \mid \varphi \in H^1_{\rm as}(\R^{6})\otimes H^1_{\rm as}(\R^6), \xi \in H^{1/2}(\R^9) \} 
\end{equation}
where, for some (arbitrary) $\mu > 0$,  $G_\mu \xi$ is the function with Fourier transform 
\begin{equation}
\widehat {G_\mu \xi}(p_1,p_2,k_1,k_2) = \sum_{i,j \in \{1,2\}} (-1)^{i+j} (h_0(p_1,p_2,k_1,k_2) + \mu)^{-1} \hat \xi(p_i+k_j, \hat p_i, \hat k_j)  
\end{equation}
and we used the notation that $\hat p_1 = p_2, \hat p_2 = p_1$ and analogously for $k$. The space $H^1_{\rm as}(\R^{6})$ denotes antisymmetric functions in $H^1(\R^{3})\otimes H^1(\R^{3})$. Note that because of the requirement  $\varphi \in H^1(\R^{12})$ the decomposition $\psi = \varphi + G_\mu \xi$ is unique. Note also that the Hilbert space under consideration consists of  functions  that are antisymmetric in the first two and last two variables, i.e., under both the exchange $p_1\leftrightarrow p_2$ and $k_1\leftrightarrow k_2$. 

For $\alpha \in \R$, the quadratic form we consider is given by
\begin{equation}
F_\alpha(\psi) = H(\varphi) - \mu \norm{\psi}^2_2 + 4 T_\mu(\xi) + 4 \alpha \norm{\xi}_2^2\,, \label{def:F} 
\end{equation}
where
\begin{equation}
H(\varphi) = \int_{\R^{12}} \left ( h_0(p_1,p_2,k_1,k_2) + \mu \right) |\hat \varphi(p_1,p_2,k_1,k_2) |^2 \, dp_1\, dp_2\, dk_1\, dk_2
\end{equation}
and $T_\mu(\xi) = \sum_{i=0}^3 \phi_i(\xi)$, with the $\phi_i$ of the form
\begin{align}
\phi_0(\xi) &= 2\pi^2 \left( \frac m{m+1} \right)^{3/2} \int |\hxi (P,p,k)|^2 \sqrt { \frac {P^2}{1+m} + p^2 + \frac {k^2}m  + \mu}\, dP\, dp\, dk \label{def:p0} \\
\phi_1(\xi) &= \int \frac {\hxi^*(p_1+k_1, p_2,k_2) \hxi(p_2+k_1, p_1,k_2)}{h_0(p_1,p_2,k_1,k_2)+ \mu}\, dp_1\, dp_2\, dk_1\, dk_2 \label{def:p1} \\
\phi_2(\xi) &= \int \frac {\hxi^*(p_1+k_1, p_2,k_2) \hxi(p_1+k_2, p_2,k_1)}{h_0(p_1,p_2,k_1,k_2)+ \mu}\, dp_1\, dp_2\, dk_1\, dk_2 \label{def:p2} \\
\phi_3(\xi) &= - \int \frac {\hxi^*(p_1+k_1, p_2,k_2) \hxi(p_2+k_2, p_1,k_1)}{h_0(p_1,p_2,k_1,k_2)+ \mu}\, dp_1\, dp_2\, dk_1\, dk_2\,. \label{def:p3}
\end{align}
We note that $F_\alpha$ is independent of the choice of $\mu>0$. The parameter $\alpha$ corresponds to the inverse scattering length; more precisely, $\alpha = -2\pi^2/a$, with $a\in (-\infty,0) \cup (0,\infty]$ the scattering length.

It was shown in \cite{Finco2010} that $T_\mu(\xi)$ is well-defined on $H^{1/2}(\R^9)$. To show stability, we need to prove that it is in fact positive. If, on the contrary, there exists a $\mu>0$ and a $\xi \in H^{1/2}(\R^9)$ such that $T_\mu(\xi)<0$, a simple scaling argument (choosing $\varphi=0$ and using the scale invariance of $F_0$) can be used to deduce that $F_\alpha$ is unbounded from below for all $\alpha\in \R$.

The functionals $\phi_0$ and $\phi_1$ also appear in a similar form in the discussion of the $2+1$ problem, and $\phi_2$ can be seen as the analogous $1+2$ term. The term $\phi_3$ has no analogue in the $2+1$ or $1+2$ systems. Note that none of the $\phi_i$ for $1\leq i \leq 3$ has a sign, and we expect that cancellations occur between them that are important for stability. In our proof below, we will first bound $\phi_0 + \phi_3$ from below by a positive quantity, which we then use to compensate separately the negative parts of $\phi_1$ and $\phi_2$. Since we shall neglect some positive terms, we cannot expect to obtain a sharp bound. In particular, whether $\mcritCur=\mcrit$, as suggested in \cite{Castin2015}, cannot be determined using this method.

\section{Main result}

For $a\in\R^3$, $b\geq 0$ and $m>0$, let  $O^m_{a,b}$ be the bounded operator on $L^2(\R^3)$ with integral kernel 
\begin{align}\nonumber
 O^m_{a,b}(p_1,p_2) & =  \left[ (p_1+a)^2 + b^2\right]^{-1/4}\left[ (p_2+a)^2 + b^2\right]^{-1/4}
\\ & \quad  \times \frac 1{  p_1^2 + p_2^2 + \frac 2{1+m} p_1\cdot p_2  + \frac {2 ( 2+m)}{(1+m)^2} a^2   + \frac {2m}{(1+m)^2}  b^2} \,.
 \label{defO}
\end{align}
Let further 
\begin{equation}\label{def:L}
\Lambda(m) = - \frac 1{2\pi^2}  \frac {1+m}{\sqrt{m}} \inf_{a\in \R^3,\, b\geq 0} \infspec O^m_{a,b}\,.
\end{equation}

\begin{theorem}
\label{thm:main}
For $m>0$ such that $\Lambda(m) + \Lambda(1/m) \leq 1$, we have 
\begin{equation}
T_\mu(\xi) 
\geq  \left( 1 -\Lambda(m) - \Lambda(1/m) \right) \sqrt{2\mu } \,\pi^2  \left(\frac {m}{m+1}\right)^{3/2} \|\xi\|_2^2 
\end{equation}
for any $\xi \in H^{1/2}(\R^9)$ and any $\mu > 0$.
\end{theorem}

This bound readily implies stability for $F_\alpha$, as the following corollary shows.

\begin{corollary}\label{cor:main}
For $m$ such that $\Lambda(m) + \Lambda(1/m) < 1$, we have
\begin{equation}
F_\alpha(\psi)\ge \begin{cases} 0 & \alpha \ge 0 \\ - \alpha^2 \left ( \frac {m+1} m \right )^3 \frac 1 {2 \pi^4 (1- \Lambda(m) - \Lambda(1/m))^{2}}   \norm{\psi}^2_2  & \alpha < 0 \end{cases} 
\end{equation}
for any $\psi \in D(F_\alpha)$. 
\end{corollary}

\begin{proof}
Without loss of generality we can assume  that $\norm{\psi}_2 = 1$. Using Theorem \ref{thm:main} and $H(\varphi)\geq 0$, we get
\begin{align*}
F_\alpha(\psi) +\mu &\ge 4 T_\mu(\xi) +  4 \alpha \norm{\xi}_2^2 \\
&\ge 4 \left [ \alpha + (1 - \Lambda(m) - \Lambda(1/m) ) \sqrt{2 \mu} \,\pi^2 \left ( \frac {m} {m+1} \right)^{3/2} \right ] \norm{\xi}^2_2 \,. \label{eq:cor:main:1}
\end{align*}
In  case  $\alpha \geq 0$ we obtain $F_\alpha(\psi) \ge - \mu$ , 
which shows the result as $\mu > 0$ was arbitrary. If $\alpha<0$,  we choose 
\begin{equation}
\mu = \alpha^2 \left ( \frac {m+1} m \right )^3 \frac 1 {2 \pi^4 (1- \Lambda(m) - \Lambda(1/m))^{2}} \,, 
\end{equation}
which yields the desired result. 
\end{proof}

We thus proved stability as long as $\Lambda(m)+\Lambda(1/m) < 1$. To investigate the implication on $m$, let us first check what happens for $a=0$ and $b=0$. An explicit calculation following \cite{correggi2012stability} shows that  
\begin{align}\nonumber 
\bar\Lambda(m) &:= - \frac 1{2\pi^2}  \frac {1+m}{\sqrt{m}} \infspec O^m_{0,0}
\\ & \,\, =\frac 2 \pi (1+m)^2  \left( \frac 1{\sqrt{m}} - \sqrt{2+m} \arcsin \left( \tfrac 1{1+m} \right) \right)\label{def:bL}
\end{align}
which satisfies $\bar \Lambda(m) + \bar\Lambda(1/m) < 1$ for $0.139 \lesssim m \lesssim 7.189$. This range of masses is the largest possible for which our approach can show stability. 

While we do not know whether $\Lambda(m) = \bar \Lambda(m)$, we shall  give in Section~\ref{sec:lambda} a rough upper bound on $\Lambda(m)$ which shows that 
 $\Lambda(m)+\Lambda(1/m) < 1$ for $\lowerBound \lesssim m \lesssim \upperBound$.

\section{Proof of Theorem \ref{thm:main}}

We shall split the proof into several steps. 

\subsection{Bound on $\phi_3$}

We shall rewrite $\phi_3$ in \eqref{def:p3} using center-of-mass and relative coordinates for each of the  pairs $(p_1,k_1)$ and $(p_2,k_2)$. With $P_1 = p_1+ k_1$, $q_1 = \frac m{1+m} p_1 - \frac 1{1+m} k_1$, $P_2 = p_2 + k_2$ and $q_2 = \frac m{1+m} p_2 - \frac 1{1+m} k_2$, we have 
\begin{align}\nonumber
 \phi_3(\xi) &= - \int dP_1\, dP_2\, dq_1\, dq_2 \\ & \qquad\quad \times \frac {\hxi^*(P_1, \frac {P_2}{1+m}  + q_2, \frac {mP_2}{1+m}-q_2) \hxi(P_2, \frac {P_1}{1+m}+ q_1 , \frac {mP_1}{1+m} -q_1)}{ \frac {1}{1+m} \left( P_1^2 + P_2^2 \right) + \frac {1+m}{m} \left(q_1^2 + q_2^2\right) + \mu}  \,.
\end{align}
By completing the square, we can write, for any positive function $w$, 
\begin{align}\nonumber
\phi_3(\xi) & =  \int \frac{ dP_1\, dP_2\, dq_1\, dq_2}{w(q_2,P_1,P_2) w(q_1,P_2,P_1)}  \\ & \qquad \times  \frac {  \frac 12 \left| \chi_w(q_2,P_1,P_2) -  \chi_w(q_1, P_2, P_1) \right|^2 - \left|   \chi_w(q_2,P_1,P_2)  \right|^2  }{ \frac {1}{1+m} \left( P_1^2 + P_2^2 \right) + \frac {1+m}{m} \left(q_1^2 + q_2^2\right) + \mu } \label{rewr3}
\end{align}
where we denote $\chi_w(q,P_1,P_2) = \hxi(P_1, \frac {P_2}{1+m}  + q, \frac {m P_2} {1+m}-q) w(q,P_1,P_2)$. 
We shall choose
\begin{equation}
w(q,P_1,P_2) = q^2 + \lambda^2 \left( \tfrac m{(1+m)^2} \left( P_1^2 + P_2^2 \right) + \tfrac m{1+m} \mu \right)
\end{equation}
for some  constant $\lambda\geq 0$. The first term in the numerator on the right side of \eqref{rewr3} is manifestly positive. Performing the integration over $q_1$, the integral over the  second term  equals
\begin{align}\nonumber
 \int  & dP_1\, dP_2\, dq_2 \,   \left( -    \frac {2\pi^2 m}{1+m} \right)  \left| \hxi(P_1, \tfrac 1{1+m} P_2 + q_2, \tfrac m{1+m}P_2-q_2) \right|^2 \\ &  \times \frac {   q_2^2 + \lambda^2 \left( \tfrac m{(1+m)^2} \left( P_1^2 + P_2^2 \right) + \tfrac m{1+m} \mu \right)  }{ \lambda \sqrt{  \frac m{(1+m)^2} \left( P_1^2 + P_2^2 \right) + \frac m{1+m}\mu }  + \sqrt  { q_2^2  + \frac {m}{(1+m)^2} \left( P_1^2 + P_2^2 \right) + \frac m{1+m} \mu} }\,.  
\end{align}

Let us compare this latter expression with $\phi_0$ in \eqref{def:p0}, which can be rewritten as 
\begin{align}\nonumber
\phi_0(\xi) & =  \frac {2\pi^2m}{m+1} \int |\hxi (P_1,\tfrac 1{1+m} P_2+q_2, \tfrac m{1+m} P_2 - q_2)|^2 \\ & \qquad \qquad\quad  \times \sqrt { q_2^2 +\frac m{(1+m)^2} \left(P_1^2+P_2^2\right)  +  \frac m {1+m} \mu}\, dP_1\, dP_2\, dq_2\,.
\end{align}
For $0\leq \lambda\leq 1$, one readily checks that 
\begin{align}\nonumber
& L_\lambda(P_1,P_2,q) \\ \nonumber& :=  \sqrt { q^2 +\frac m{(1+m)^2} \left(P_1^2+P_2^2\right)  +  \frac m {1+m} \mu}   \\ & \quad\  - \frac {   q^2 + \lambda^2 \left( \tfrac m{(1+m)^2} \left( P_1^2 + P_2^2 \right) + \tfrac m{1+m} \mu \right)   }{ \lambda \sqrt{  \frac m{(1+m)^2} \left( P_1^2 + P_2^2 \right) + \frac m{1+m} \mu }  + \sqrt  { q^2  + \frac {m}{(1+m)^2} \left( P_1^2 + P_2^2 \right) + \frac m{1+m} \mu} } 
\end{align}
is non-negative. What we have shown here is that
\begin{align}\nonumber 
& \phi_0(\xi) + \phi_3(\xi) \\ & \geq   \frac {2\pi^2m}{m+1} \int |\hxi (P_1,\tfrac 1{1+m} P_2+q, \tfrac m{1+m} P_2 - q)|^2 L_\lambda (P_1,P_2,q) \, dP_1\, dP_2\, dq \label{bd1}
\end{align}
for any $\lambda \geq 0$.

Note that for $\lambda^2=1/2$, $L_\lambda$ takes the simple form
\begin{equation}\label{def:L2}
L_{1/\sqrt{2}}(P_1,P_2,q) =  \frac 1 {\sqrt 2} \sqrt {\tfrac m{(1+m)^2} \left( P_1^2 + P_2^2 \right) + \tfrac m{1+m} \mu }
\end{equation}
and is, in particular, independent of $q$.

\subsection{Bound on $\phi_1$}

For the term $\phi_1$ in \eqref{def:p1}, we shall switch to center-of-mass and relative coordinates for the particles $(p_1,p_2,k_1)$. With $P= p_1+p_2+k_1$,  $q_1 =\frac {1+m}{2+m} p_1 - \frac 1{2+m} (p_2 + k_1)  $ and $q_2 = \frac {1+m}{2+m} p_2 - \frac 1{2+m} (p_1 + k_1)$, as well as $k=k_2$ for short, we have 
\begin{align}\nonumber
\phi_1(\xi)  &=    \frac m{1+m}  \int dP \, dq_1\, dq_2\, dk \\ & \qquad\qquad\quad \times \frac {\hxi^*(\frac {1+m}{2+m} P-q_2 , \frac P{2+m}+q_2 ,k) \hxi(\frac{1+m}{2+m} P-q_1, \frac P{2+m} + q_1 ,k)}{  q_1^2 + q_2^2 + \frac 2{1+m} q_1\cdot q_2  + \frac m{(1+m)(2+m)} P^2 + \frac 1{1+m} k^2 + \frac m{1+m} \mu} \,.
\end{align}
Defining
\begin{equation}\label{def:ell}
 \ell_\lambda(q,P,k)  = L_\lambda (\tfrac {1+m}{2+m} P - q, \tfrac P{2+m} + q +k , \tfrac {mq}{1+m}  + \tfrac {mP}{(1+m)(2+m)}  - \tfrac k{1+m} )
\end{equation}
our aim is to obtain a lower bound on  the operator on $L^2(\R^3)$ with integral kernel 
\begin{align}\nonumber 
 & \ell_\lambda (q_1,P ,k)^{-1/2} \ell_\lambda(q_2,P ,k)^{-1/2}  \\ & \times \frac 1{  q_1^2 + q_2^2 + \frac 2{1+m} q_1\cdot q_2  + \frac m{(1+m)(2+m)} P^2 + \frac 1{1+m} k^2 + \frac m{1+m} \mu} 
\end{align}
for suitable $\lambda$, uniformly in the fixed parameters $P$ and $k$. 

Let us take $\lambda^2=1/2$ for simplicity, in which case we have
\begin{equation}
\ell_{1/\sqrt 2}(q, P ,k) =  \frac{\sqrt{m}}{1+m} \sqrt{
 \left( q  + \tfrac 12 k  - \tfrac m{2(2+m)} P \right)^2+ \tfrac 14 \left(    P + k \right)^2
+ \tfrac {1+m}{2} \mu }\,.
\end{equation}
Note also that 
\begin{align}\nonumber
& \frac m{(1+m)(2+m)} P^2 + \frac 1{1+m} k^2  \\ & =   \frac {2m}{(1+m)^2}  \left[ \frac {2+m} m    \left(  \tfrac 12 k  - \tfrac m{2(2+m)} P  \right)^2+ \frac 1 4  \left(    P + k \right)^2 \right]\,.
\end{align}
With
\begin{equation}
a =  \tfrac 12 k  - \tfrac m{2(2+m)} P \quad, \quad b^2 = \tfrac 14   \left(    P + k \right)^2
  + \tfrac {1+m}2\mu
\end{equation}
our task is thus to find a lower bound on the operator with integral kernel $ \frac {1+m}{\sqrt{m}} O^m_{a,b}(q_1,q_2)$, defined in \eqref{defO}. The best lower bound equals $-2 \pi^2 \Lambda(m)$, by definition.

To summarize, what we have shown here is that 
\begin{equation}
 \phi_1(\xi)  \geq   -\Lambda(m)  \frac {2\pi^2m}{m+1}  \int |\hxi (\tfrac {1+m}{2+m} P -q , \tfrac P{2+m}+q  ,k) |^2 \ell_{1/\sqrt{2}}(q, P ,k) \, dP\, dq \, dk\,.
\end{equation}
Using \eqref{def:ell}, a simple change of variables shows that this is equivalent to 
\begin{align}\nonumber
& \phi_1(\xi)  \\ & \geq  
  -\Lambda(m)  \frac {2\pi^2m}{m+1} \int |\hxi (P_1,\tfrac {P_2}{1+m}  +q, \tfrac {mP_2} {1+m} - q)|^2 L_{1/\sqrt{2}}(P_1,P_2,q)
 \, dP_1\, dP_2\, dq \,.\label{bd2}
\end{align}

\subsection{Bound on $\phi_2$}

In exactly the same way we proceed with $\phi_2$ in \eqref{def:p2}, which 
we rewrite  as
\begin{align}\nonumber
&\phi_2(\xi)  \\ \nonumber &=  \frac m{1+m}  \int  d P\, dq_1\, dq_2\, dp \\ & \qquad\qquad\quad \times  \frac {\hxi^*( \frac {1+m}{1+2m} P- q_2, p,q_2 + \frac {mP}{1+2m} ) \hxi( \frac{1+m}{1+2m} P-q_1, p,q_1 + \frac {mP}{1+2m} )}{  q_1^2 + q_2^2 +\frac {2m}{1+m} q_1\cdot q_2 + \frac m{(1+m)(1+2m)}  P^2  + \frac m{1+m} p^2  + \frac m{1+m} \mu}\,.
\end{align}
If we now define
\begin{equation}
 \tilde\ell_\lambda(q,P,p)  = L_\lambda (\tfrac {1+m}{1+2m} P -q, p + q + \tfrac {mP}{1+2m} , \tfrac {mp}{1+m}  - \tfrac q{1+m} -  \tfrac {mP}{(1+m)(1 + 2m)} )
\end{equation}
we need a lower bound on  the operator on $L^2(\R^3)$ with integral kernel 
\begin{align}\nonumber 
& \tilde\ell_\lambda (q_1,P,p)^{-1/2}  \tilde\ell_\lambda(q_2,P, p)^{-1/2}  \\ & \times \frac 1{  q_1^2 + q_2^2 + \frac {2m}{1+m} q_1\cdot q_2  + \frac m{(1+m)(1+2m)} P^2 + \frac m{1+m} p^2 + \frac m{1+m} \mu}
\end{align}
for fixed $P$ and $p$. 
By proceeding as in the previous subsection, one readily checks that, for $\lambda^2=1/2$, its best lower bound is $-2 \pi^2\Lambda(1/m)$, with $\Lambda$ defined in \eqref{def:L}. In particular, we have 
\begin{align}\nonumber
& \phi_2(\xi)  \\ & \geq  -\Lambda(1/m)  \frac {2\pi^2m}{m+1} \int |\hxi (P_1,\tfrac {P_2}{1+m} +q, \tfrac {mP_2}{1+m}  - q)|^2 L_{1/\sqrt{2}}(P_1,P_2,q)
 \, dP_1\, dP_2\, dq  \,.\label{bd3}
\end{align}

\subsection{Combining  above bounds}

By combining the bounds \eqref{bd1}, \eqref{bd2} and \eqref{bd3} from the previous three subsections, we obtain 
\begin{align}\nonumber
T_\mu(\xi) &=  \sum_{j=0}^3 \phi_j(\xi)  \\ &  \geq \left( 1 -\Lambda(m) - \Lambda(1/m) \right) \frac {2\pi^2m}{m+1}  \nonumber \\ & \quad  \times \int |\hxi (P_1,\tfrac 1{1+m} P_2+q, \tfrac m{1+m} P_2 - q)|^2 L_{1/\sqrt{2}}(P_1,P_2,q)  \, dP_1\, dP_2\, dq 
\end{align}
with $L_{1/\sqrt{2}}$ defined in \eqref{def:L2}. 
In the case  $\Lambda(m) + \Lambda(1/m) \leq 1$, we can further use  $L_{1/\sqrt{2}}(P_1,P_2,q) \geq \sqrt{ m\mu / (2(1+m))}$ for a lower bound. This completes the proof of Theorem~\ref{thm:main}.
 \hfill\qed

\section{Bound on $\Lambda(m)$}\label{sec:lambda}

Note that $\Lambda(m)\geq \bar\Lambda(m)$. To obtain an upper bound, we use the Schur test. We first drop the positive part of the operator with integral kernel 
\begin{equation}
k(p_1, p_2) =  \left[  p_1^2 + p_2^2 + \frac 2{1+m} p_1\cdot p_2  + \frac {2 ( 2+m)}{(1+m)^2} a^2   + \frac {2m}{(1+m)^2}  b^2 \right]^{-1}\,.
\end{equation}
It follows from \cite[Lemma~3]{MoserSeiringer2017} that the negative part of this operator has the integral kernel
\begin{align*}
k_-(p_1,p_2) &= \frac {-k(p_1,p_2) + k(p_1,-p_2)}2 \\
&=  \frac 2 {1+m} \frac {p_1\cdot p_2}{  \left[ p_1^2 + p_2^2   + \frac {2 ( 2+m)}{(1+m)^2} a^2   + \frac {2m}{(1+m)^2} b^2  \right]^2 - \frac {4\left(p_1\cdot p_2\right)^2}{(1+m)^2}  } \,.
\end{align*}
By applying the Cauchy-Schwarz inequality, we obtain, for any positive function $h$ on $\R^3$ (possibly depending on $a$ and $b$) 
\begin{align}\nonumber
 \Lambda(m) & \leq  \frac 1{\pi^2\sqrt{m}}  \sup_{p_1,a,b}  \int_{\R^3} \frac{h(p_1)}{h(p_2)}
 \frac {|p_1\cdot p_2|}{  \left[ p_1^2 + p_2^2   + \frac {2 ( 2+m)}{(1+m)^2} a^2   + \frac {2m}{(1+m)^2} b^2  \right]^2 - \frac {4\left(p_1\cdot p_2\right)^2}{(1+m)^2}  }  
\\ & \qquad \qquad \qquad \qquad \times \left[ (p_2+a)^2 + b^2\right]^{-1/2} dp_2 \,.
\end{align}

By monotonicity, we can set $b=0$, i.e,
\begin{equation}
  \Lambda(m) \!\leq\! \frac 1{\pi^2\sqrt{m}}  \sup_{p_1,a}  \int_{\R^3}\frac{h(p_1)}{h(p_2)}
   \frac {|p_1\cdot p_2|}{  \left[ p_1^2 + p_2^2   + \frac {2 ( 2+m)}{(1+m)^2} a^2   \right]^2 -\frac {4\left(p_1\cdot p_2\right)^2}{(1+m)^2}   }  
\left| p_2\!+a \right|^{-1} dp_2 .
\end{equation}
We shall choose $h$ to be even, i.e., $h(p)=h(-p)$, in which case we can symmetrize to get
\begin{align} \nonumber 
 \Lambda(m) & \leq \frac 1{\pi^2\sqrt{m}}  \sup_{p_1,a} \int_{\R^3} \frac{h(p_1)}{h(p_2)}
   \frac {|p_1\cdot p_2|}{  \left[ p_1^2 + p_2^2   + \frac {2 ( 2+m)}{(1+m)^2} a^2   \right]^2 - \frac {4\left(p_1\cdot p_2\right)^2}{(1+m)^2}  }   
  \\ & \qquad \qquad\qquad \qquad \times \frac 12 \left( \frac 1{ \left| p_2+a \right|} +  \frac 1{\left| p_2-a \right|}\right) dp_2 \nonumber \\ & 
 \leq \frac 1{\pi^2\sqrt{m}}  \sup_{p_1,a}  \int_{\R^3} \frac{h(p_1)}{h(p_2)} 
   \frac {|p_1\cdot p_2|}{  \left[ p_1^2 + p_2^2   + \frac {2 ( 2+m)}{(1+m)^2} a^2   \right]^2 - \frac {4\left(p_1\cdot p_2\right)^2}{(1+m)^2}  }   \nonumber
 \\ & \qquad \qquad\qquad \qquad \times \sqrt { \frac {p_2^2 + a^2}{ \left( p_2^2 + a^2\right)^2 - 4 (p_2\cdot a)^2}} dp_2\,.
\end{align}
To maximize the right side, $a$ wants to be parallel to $p_1$, i.e., $a=\kappa p_1$ for $\kappa\in \R$. This is a direct consequence of \cite[Lemma 5]{MoserSeiringer2017}. We shall choose $h(p)=|p|$. By scale invariance we can set $|p_1|=1$. We then obtain
\begin{align}\nonumber
\Lambda(m)  &\leq \frac 4{\pi\sqrt{m}}  \sup_{\kappa \in \R}  \int_{0}^1 dt \int_0^\infty dr 
   \frac {r^2 t }{  \left[ 1 + r^2   + \frac {2 ( 2+m)}{(1+m)^2} \kappa^2   \right]^2 - \frac {4r^2 t^2} {(1+m)^2}   }  
 \\ & \qquad \qquad \qquad \qquad \qquad  \quad \times  \sqrt { \frac {r^2 + \kappa^2}{ \left( r^2 + \kappa^2\right)^2 - 4 \kappa^2 r^2 t^2 }} \,. \label{lm2}
\end{align}

We further bound $t\leq 1$ in the denominator of the first term in the integrand in \eqref{lm2}, and use that 
\begin{equation}
 \left[ 1 + r^2   + \frac {2 ( 2+m)}{(1+m)^2} \kappa^2   \right]^2\! - \frac {4 r^2}{(1+m)^2}   
  \geq  \frac {m(m+2)}{(1+m)^2} \left[ 1 + r^2   + \frac {2 \sqrt{2+m} }{(1+m)\sqrt{m}} \kappa^2   \right]^2 .
\end{equation}
Since
\begin{equation}
 \int_{0}^1 dt  \, t 
 \sqrt { \frac {r^2 + \kappa^2}{ \left( r^2 + \kappa^2\right)^2 - 4 \kappa^2 r^2 t^2 }}  = \frac 1{2r^2} \sqrt{r^2 +\kappa^2} \min\{ 1, r^2/\kappa^2\}
\end{equation}
we therefore get 
\begin{equation}\label{lblr}
\Lambda(m)  \leq \frac 2{\pi}  \frac {(1+m)^2}{m^{3/2}(m+2)} \sup_{\kappa \in \R}   \int_0^\infty dr 
   \frac { \sqrt{r^2 +\kappa^2} }{  \left[ 1 + r^2   + \frac {2 \sqrt{2+m} }{(1+m)\sqrt{m}} \kappa^2   \right]^2     } \,.  
\end{equation}
We define $c_m = 2 \sqrt{2+m}/((1+m) \sqrt m)$. After explicitly doing the integral, the bound \eqref{lblr} reads
$\Lambda(m)\leq \lambda(m):=\sup_{\kappa>0} \lambda(m,\kappa)$ with 
\begin{align}\nonumber
\lambda(m,\kappa) & := \frac 1{\pi}  \frac {(1+m)^2}{m^{3/2}(m+2)}   \frac 1{1+ c_m \kappa^2} \Bigg( 1 + \frac {\kappa^2}{\sqrt{1+ c_m \kappa^2} \sqrt{1+ \kappa^2 (c_m - 1) }} \\
&\quad \times \ln \left ( \frac{ \sqrt{1+c_m \kappa^2} +\sqrt{1 + \kappa^2(c_m - 1)}}{\kappa} \right ) \Bigg) \,.
\end{align}

For our purpose it is important that $\lambda(1)\approx 0.427 < 1/2$ (see Fig.~\ref{fig:1}). 
\begin{figure}
  \includegraphics[width=0.75\textwidth]{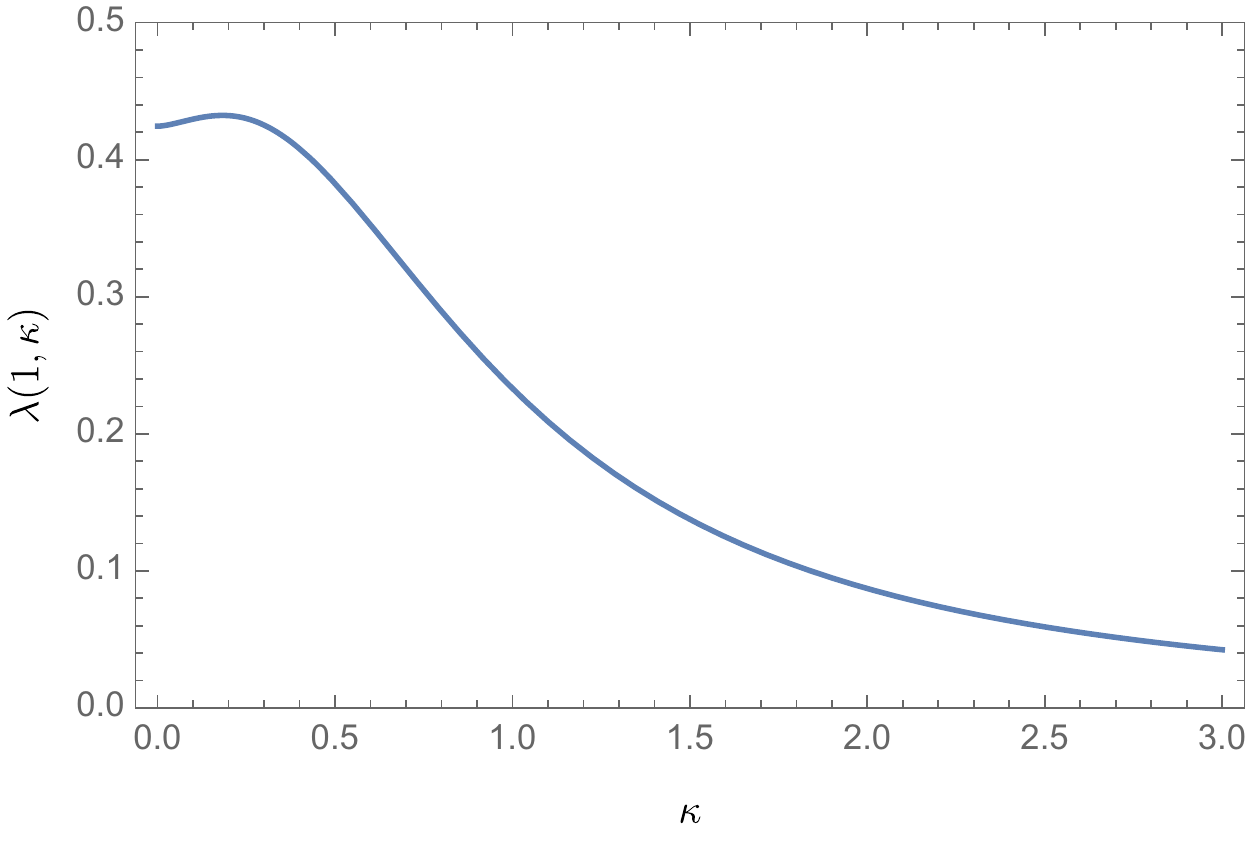}
\caption{The function $\lambda(1,\kappa)$, with $\lambda(1)=\sup_\kappa \lambda(1,\kappa) \approx 0.427 $}
\label{fig:1}       
\end{figure}
By continuity, this implies that $\Lambda(m) + \Lambda(1/m) < 1$ for a window of mass ratios around $1$. In fact, 
a numerical optimization over $\kappa$ leads to the conclusion that $\Lambda(m) + \Lambda(1/m) < 1$ whenever $\lowerBound \approx \mcritCur < m < \mcritCur^{-1} \approx \upperBound$ (see Fig~\ref{fig:2}). 
\begin{figure*}
  \includegraphics[width=0.75\textwidth]{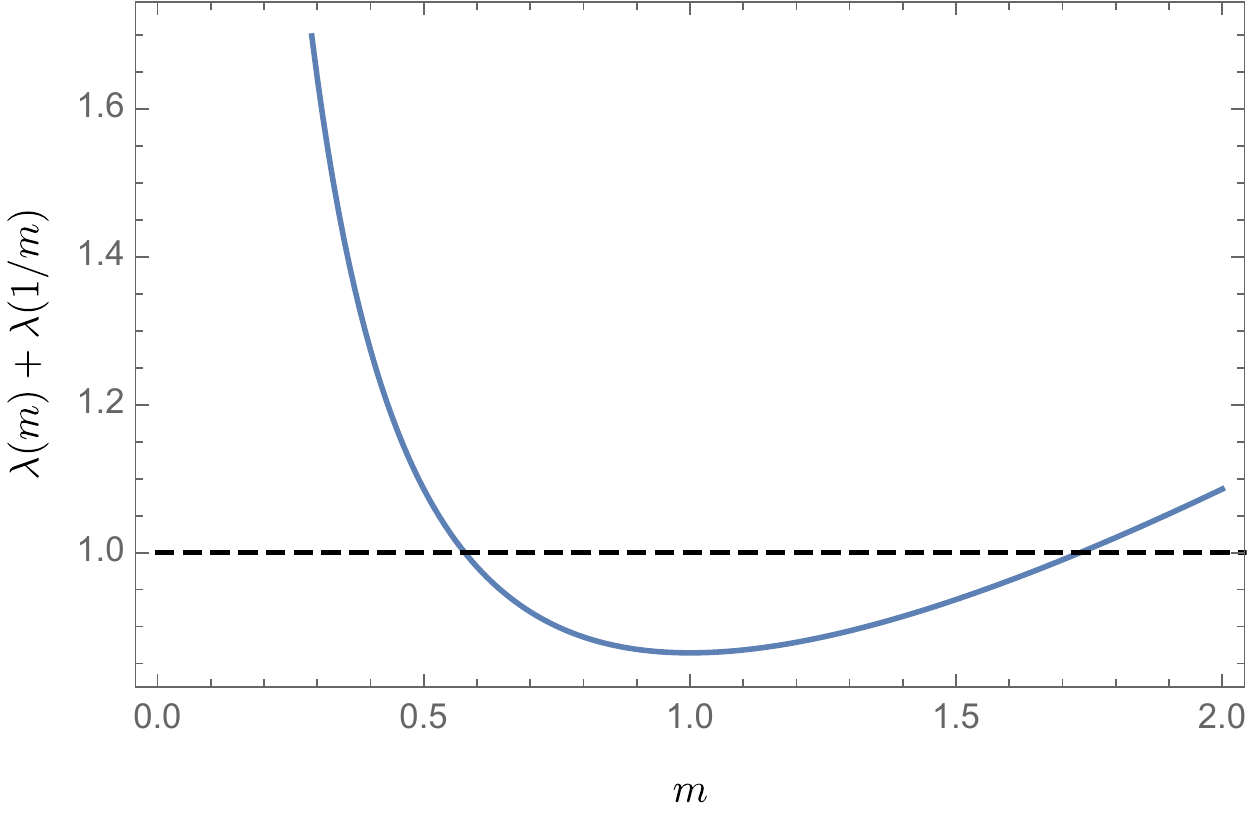}
\caption{Our upper bound on $\Lambda(m)+\Lambda(1/m)$, given by  $ \lambda(m) + \lambda(1/m)$}
\label{fig:2}       
\end{figure*}

\subsection*{Acknowledgments} 
Financial support by the European Research Council (ERC) under the European Union's Horizon 2020 research and innovation programme (grant agreement No 694227), and by the Austrian Science Fund (FWF), project Nr. P 27533-N27, is gratefully acknowledged.


\begin{thebibliography}{23}
%
\bibitem{albeverio1988solvable}
S.~Albeverio, F.~Gesztesy, R.~H\o egh-Krohn, H.~Holden, 
\newblock \emph{Solvable models in quantum mechanics}, $2^{\rm nd}$ ed., 
\newblock Amer. Math. Soc. (2004).

\bibitem{basti}
G. Basti, C. Cacciapuoti, D. Finco, A. Teta, 
{\it The three-body problem in dimension one: From short-range to contact interactions}, preprint, arXiv:1803.08358.
 
\bibitem{Becker2017}
S.~Becker, A.~Michelangeli, A.~Ottolini,
\newblock \emph{Spectral properties of the 2+1 fermionic trimer with contact
  interactions}, preprint, 
\newblock {arXiv:1712.10209}.

\bibitem{Bethe1935}
H.~Bethe, R.~Peierls,
\newblock \emph{Quantum theory of the diplon},
\newblock Proc. R. Soc. Lond. Ser. A \textbf{148},  146--156 (1935).

\bibitem{Bethe1935B}
H.~Bethe, R.~Peierls,
\newblock \emph{The Scattering of Neutrons by Protons},
\newblock Proc. R. Soc. Lond. Ser. A \textbf{149}, 176--183 (1935).

\bibitem{correggi2012stability}
M.~Correggi, G.~Dell'Antonio, D.~Finco, A.~Michelangeli, A.~Teta,
\newblock \emph{Stability for a System of N Fermions plus a different Particle
  with Zero-Range Interactions},
\newblock Rev. Math. Phys. \textbf{24}, 1250017 (2012).

\bibitem{Correggi2015}
M.~Correggi, G.~Dell'Antonio, D.~Finco, A.~Michelangeli, A.~Teta,
\newblock \emph{A Class of Hamiltonians for a Three-Particle Fermionic System
  at Unitarity},
\newblock Math. Phys. Anal. Geom. \textbf{18}, pp. 1--36
  (2015).

\bibitem{Correggi2015B}
M.~Correggi, D.~Finco, A.~Teta,
\newblock \emph{Energy lower bound for the unitary $N+1$ fermionic model},
\newblock Eur. Phys. Lett. \textbf{111},  10003 (2015).

\bibitem{DellAntonio1994}
G.~F. Dell'Antonio, R.~Figari, A.~Teta,
\newblock \emph{Hamiltonians for systems of N particles interacting through
  point interactions},
\newblock Ann. Inst. Henri Poincar\'e \textbf{60},  253--290 (1994).

\bibitem{Dimock2004}
J.~Dimock, S.~G. Rajeev,
\newblock \emph{{Multi-particle Schr{\"{o}}dinger operators with point
  interactions in the plane}},
\newblock J. Phys. A: Math. Gen. \textbf{37}, 
  9157--9173 (2004).

\bibitem{Castin2015}
S.~Endo, Y.~Castin,
\newblock \emph{Absence of a four-body Efimov effect in the $2+2$ fermionic
  problem},
\newblock Phys. Rev. A \textbf{92}, 053624 (2015).

\bibitem{Fermi1936}
E.~Fermi,
\newblock \emph{Sul moto dei neutroni nelle sostanze idrogenate},
\newblock Ric. Sci. Progr. Tecn. Econom. Naz. \textbf{7},  13--52 (1936).

\bibitem{Finco2010}
D.~Finco, A.~Teta,
\newblock \emph{{Remarks on the Hamiltonian for the Fermionic Unitary Gas
  model}},
\newblock Rep. Math. Phys. \textbf{69},  131--159 (2012).

\bibitem{Ulrich2017}
M.~Griesemer, U.~Linden,
\newblock \emph{Stability of the two-dimensional Fermi polaron},
\newblock preprint, arXiv: 1709.02691, Lett. Math. Phys. (in press).

\bibitem{massignan2014polarons}
P.~Massignan, M.~Zaccanti, G.~M. Bruun,
\newblock \emph{Polarons, dressed molecules and itinerant ferromagnetism in
  ultracold Fermi gases},
\newblock Rep. Prog. Phys. \textbf{77},  034401 (2014).

\bibitem{michelangeli2016}
A.~Michelangeli, P.~Pfeiffer,
\newblock \emph{Stability of the $(2+2)$-fermionic system with zero-range
  interaction},
\newblock J. Phys. A: Math. Theor. \textbf{49}, 105301 (2016).

\bibitem{Minlos2011}
R.~Minlos,
\newblock \emph{On point-like interaction between $n$ fermions and another particle},
\newblock Moscow Math. J. \textbf{11}, 113--127 (2011).


\bibitem{Minlos2012}
R.~Minlos,
\newblock \emph{On Pointlike Interaction between Three Particles: Two Fermions
  and Another Particle},
\newblock ISRN Math. Phys. \textbf{2012},  230245 (2012).

\bibitem{Minlos2014}
R.~Minlos,
\newblock \emph{On Pointlike Interaction between Three Particles: Two Fermions
  and Another Particle II},
\newblock Moscow Math. J. \textbf{14},  617--637 (2014).

\bibitem{Minlos2014B}
R.~Minlos,
\newblock \emph{A system of three quantum particles with point-like
  interactions},
\newblock Russian Math. Surveys \textbf{69},  539--564 (2014).

\bibitem{MoserSeiringer2017}
T.~Moser, R.~Seiringer,
\newblock \emph{Stability of a fermionic $N+ 1$ particle system with point
  interactions},
\newblock Commun. Math. Phys. \textbf{356},  329--355 (2017).

\bibitem{sherm}
M. K. Shermatov, {\it Point Interaction Between Two Fermions and One Particle of a Different Nature}, Theor. Math. Phys. {\bf 136}, 1119--1130 (2003).

\bibitem{Thomas1935}
L.~H. Thomas,
\newblock \emph{{The interaction between a neutron and a proton and the
  structure of $H^3$}},
\newblock Phys. Rev. \textbf{47},  903--909 (1935).

\bibitem{wigner1933streuung}
E.~Wigner,
\newblock \emph{Über die Streuung von Neutronen an Protonen},
\newblock Z. Phys. \textbf{83},  253--258 (1933).

\bibitem{ZwergerBCS}
W.~Zwerger, ed., 
\newblock \emph{The BCS--BEC Crossover and the Unitary Fermi Gas}, Springer Lecture Notes in Physics {\bf 836} (2012).

\end{thebibliography}
\end{document}